\documentclass{article}
\usepackage{amsmath,amssymb,bm,bbm,mathrsfs,amscd}
\usepackage{calc}
\usepackage{color, subfig}
\usepackage{amsthm}
\usepackage{dsfont}
\usepackage{graphicx}
\usepackage{epstopdf}
\usepackage{epsfig}
\usepackage{tikz}
\usetikzlibrary {positioning}

\newtheorem{theorem}{Theorem}

\newtheorem{proposition}[theorem]{Proposition}
\newtheorem{lemma}[theorem]{Lemma}
\newtheorem{corollary}[theorem]{Corollary}
\newtheorem{remark}[theorem]{Remark}
\newcommand{\abs}[1]{\lvert#1\rvert}

\newcommand{\ba}{\begin{array}}
\newcommand{\ea}{\end{array}}

\newcommand{\be}{\begin{equation}}
\newcommand{\ee}{\end{equation}}
\newcommand{\eps}{\varepsilon}

\newcommand{\mc}{\mathcal}

\newcommand{\1}{\mathbbm{1}}
\newcommand{\E}{\mathbb{E}}

\newcommand{\N}{\mathbb{N}}

\renewcommand{\P}{\mathbb{P}}

\def\1{\mathds{1}}

\def\N{\mathbb{N}}
\def\E{\mathbb{E}}

\def\P{\mathbb{P}}

\title{Diffusion of innovation in large scale graphs}
\author{Fabio Fagnani\thanks{Fabio Fagnani is with the Department of Mathematical Science, Polytechnic University of Turin, Torino, Italy
        {\tt\small fabio.fagnani@polito.it}} \and Lorenzo Zino\thanks{Lorenzo Zino is with the Department of Mathematical Science, Polytechnic University of Turin,
        Torino, Italy
        {\tt\small lorenzo.zino@unito.it}}%
}

\begin{document}
\maketitle

\begin{abstract}
Will a new smartphone application diffuse deeply in the population or will it sink into oblivion soon? To predict this, we argue that common models of spread of innovations based on cascade dynamics or epidemics may not be fully adequate. Therefore we propose a novel stochastic network dynamics modeling the spread of a new technological asset, whose adoption is based on the word-of-mouth and the persuasion strength increases the more the product is diffused. In this paper we carry on an analysis on large scale graphs to show off how the parameters of the model, the topology of the graph and, possibly, the initial diffusion of the asset, determine whether the spread of the asset is successful or not. In particular, by means of stochastic dominations and deterministic approximations, we provide some general results for a large class of expansive graphs. Finally we present numerical simulations trying to expand the analytical results we proved to even more general topologies.

\end{abstract}

\section{Introduction}

In this paper we extend the analysis of a novel network dynamics that models the diffusion of the adoption of a new technological feature in a population, to the case of a general network topology. According to this dynamics, introduced in earlier work \cite{mtns}, we imagine that agents (individuals) are connected through a network and update their state (representing whether they use the new asset or not) at random times, according to two different mechanisms: a gossip persuasion mechanism which pushes adoption by contacting a neighbor already possessing the new technology, and a spontaneous regression mechanism where an agent autonomously abandons the new technology.

In many papers (\cite{valente}-\cite{young}), the spread of innovations and of new behaviors is modeled and studied under the assumption that each individual tends to maximize his or her payoff in a coordination game within his or her neighborhood. This hypothesis seems to be very realistic as far as we are concerned with what we might call ``big choices'', such as the terms of economic contracts \cite{young} or the choice of an operating system \cite{montanari}. In such cases, a wrong decision can be very costly for the one who took it, therefore it is more than reasonable that an individual contacts many of his or her friends/colleagues (i.e. the neighbors) before taking the ``big choice''. In this work we want to focus on those we might call ``light choices'', for example to download an application on the mobile phone or to join an online community/social network. In such cases, negative consequences of the choice are usually mild, hence we can assume that an individual takes his or her choice after a pairwise interaction with one of its neighbors who has already taken it (a recent survey \cite{nie} supports our hypothesis highlighting the centrality of the world-of-mouth in the spread of assets), instead of an analysis of its whole neighborhood.

The main novelty of this model, with respect to classical epidemic models (see \cite{and}, \cite{pastor}), lays in the fact that the strength of the gossip persuasion depends on the global diffusion that the new item already has in the population. Neighbors are thus sort of media through which an agent gets aware of the existence of this new item, but it is the size of its diffusion in the population that determines its attraction for the new user. Instead of a diffusion channel, the gossip mechanism is rather a learning channel. The driving force of this learning is the value of the good itself which depends on the number of individuals already possessing it, a phenomenon that in economy is known as a positive externalities effect \cite{Katz}. It should also be remarked that our model differs, substantially, from network dynamical model with neighborhood effects (typically cascade dynamics) where the driving force depends on the number of neighbors possessing the new feature. For technological items whose value for a unit depends on fractions of neighbors (friends) already adopting it, these last models are the right one to consider. Instead for items like a new application for smart-phone, a new program for PCs, their value and their reliability strictly depends on their diffusion in the whole population. For such context, our model turns out to be more appropriate.

The presence of the abandoning mechanisms has various motivations. From the applicative point of view, it may model a tendency to abandon technologies which have a maintenance cost or also a sort of limitation to the time period during which a unit can influence their neighbors. Our analysis will however cover also, as a limit case, the situation when this term is not present. Mathematically, the case when only the persuasion mechanism is present is not particularly interesting as in this case adoption will diffuse to the entire population independently on the strength of the mechanism and on the topology of the network. 

Formally, our model is a jump Markov process \cite{levin} on a space whose dimension grows exponentially in the number of individuals. As for the SIS model \cite{draief}, the regression mechanism induces an absorbing state where no agent has the asset and standard probabilistic arguments (such as Borel-Cantelli lemma) allow to conclude that with probability $1$, the system will be, in finite time, in this absorbing configuration. The key point is thus to analyze the maximum level of diffusion of the asset in the population and its persistence in time before the absorbing event. 

In the SIS model, we witness the presence of two different regimes which depend of the strength of the contagion mechanism with respect to the spontaneous regression. If this strength is too small, the epidemic quickly decades to the absorbing state, while if it is above a certain threshold, the epidemic expands and remains persistent in the population, for a time scaling exponentially in the size of the population (see \cite{draief}). In our model, besides these two regimes, we witness, in many cases, the presence of a third intermediate regime where the behavior (either fast extinction or long persistence) strictly depends on the initial condition, namely the original fraction of users of the asset in the population. If the original diffusion of the new item in the population is below a certain threshold, the item will not be able to spread, while if it is above, the persuasion mechanism will be able to push towards a wide and persistent diffusion. In \cite{mtns} we proved the existence of this regime for some easy mean field network topologies (complete graph communities). A different mean field model with the same driving mechanism was considered in \cite{Easley} (see Chapter 17) and analogous dependences on the initial fraction of users was found. In this paper we want to extend our results to very general graphs, relating the presence of this regime to some connectivity features of the graph.

The main results of this paper are Theorem \ref{absorbing}, originally presented in \cite{mtns}, and Theorem \ref{absorbing general}, which generalized the previous to the case of a general interaction graph. In this case, the presence of the various regimes is not only due to the strength of the gossip persuasion mechanism (as in Theorem \ref{absorbing}) but, also, to the topology of the interactions (measured through the degrees, the spectral radius of the adjacency matrix and the bottleneck ratio). Even if results in Theorem \ref{absorbing general} are not exhaustive as in the mean field case, nevertheless, they are sufficient to prove the existence of the intermediate regime where the behavior strongly depends on the initial condition, for a number of relevant families of graphs including random Erd\H{o}s-R\'{e}nyi graphs and random graphs with prescribed (bounded) degree distribution.

In the remaining part of this Section, we present a formal description of the proposed model, as well all notation and concepts used throughout the paper. Section \ref{section mean field} is devoted to the mean field approximation (i.e. the analysis on a complete graph) and Section \ref{section general} to the analysis on general graphs. Finally, Section \ref{section specific} contains the application of results proved in Section \ref{section general} to a number of specific families of graphs and a number of simulation results corroborating the analytical results and actually showing that also in less connected family of graphs, such phase transition phenomena do take place even if our analysis is not sufficient to prove their existence.

\subsection{Description of the model}\label{section description}
Let $G =(V, E)$ be a directed graph with a finite set of nodes $V$, called \emph{agents}, and set of edges $ E\subseteq  V\times  V$. Put $N=|V|$.
The presence of an edge $(v,w)$ has to be interpreted in the sense that agent $v$ is influenced by agent $w$. Let $\mc N_v$ be the set of the (out-)neighbors of $v$, namely
$$
\mc N_v=\left\{ w\in V : (v,w)\in  E\right\}\,,
$$
while  $d_v=|\mc N_v|$ denotes the (out-)degree of $v$.
Notice that $\mc N_v$ is the set of agents who influence $v$. 
Agents are described by their state, in particular $X_v(t)=1$ if agent $v$ has the asset at time $t$, otherwise $X_v(t)=0$. States are assembled in a vector $X(t)$, called the {\it state configuration} of the system at time $t$. $\delta_v$ denotes the configuration where the agent $v$ is in state $1$, while all other agents are in state $0$ and $x\1$ denotes the pure configuration where all nodes have state $x$. Given ${\bf y}\in \{0,1\}^{V}$, we define $z({\bf y})=N^{-1}\sum_i y_i\in [0,1]$ to be the fraction of $1$'s among the components of the vector ${\bf y}$. 

Dynamics is defined as follows: nodes and edges are equipped with independent Poisson clocks. Agents activate at rate $1$ (this choice is just for the sake of simplicity), while edges activate at rate $\beta\bar{d}^{-1}$ where $\bar{d}$ is the average degree of the graph (this rescaling with respect to $\bar d$ is useful in presenting our large scale results). 
When agent $v$ or edge $(v,w)$ activates, agent $v$ will have the possibility to revise its state according to:
\begin{itemize}
\item {\bf Persuasion by gossip}: Fix a function $\phi:[0,1]\to [0,1]$. If the edge $(v,w)$ activates at time $t$,  $X_v(t)=0$, $X_w(t)=1$ and the fraction of $1$'s in the population is $z$, then agent $v$ updates its state to $1$ with probability $\phi(z)$.
\item {\bf Spontaneous regression}: If the agent $v$ activates at time $t$ and  $X_v(t)=1$, its state spontaneously changes to $0$.
\end{itemize}
Formally, $X(t)$ is a jump Markov process on $\{0,1\}^V$ whose non-zero transition rates from $X(t)=\bf y$ are:
\be\label{transitions}\left\{\ba{l}
\lambda_{{\bf y}, {\bf y}+\delta_v}=\beta\bar{d}^{-1}(1-y_v)\sum\limits_{w\in N_v}y_w\phi(z({\bf y}))\\
\lambda_{{\bf y}, {\bf y}-\delta_v}=y_v.
\ea\right.\ee
Notice that when $\phi$ is constant, this model reduces to the classical SIS model \cite{draief}. The main feature and novelty of this model is the fact that, when the function $\phi$ is instead not constant, the gossip dynamics is affected by the global distribution of the state in the population of agents. In this model agents influence each other through two ``information channels'': the one coming from the graph $ G$ and the another one due to the global pressure of the population state. These two channels are coupled through the persuasion mechanism described above.

In this paper we will exclusively consider $\phi$ non-decreasing since we are interested in the effect of positive externalities and concave, since it seems reasonable to consider the persuasion effect to increase lower than linearly as the asset spreads (trivially an increase of $1$ is certainly felt more when starting from $10$ than when starting from $1000$). In \cite{mtns} we called these assumptions \emph{standard assumptions}. For the sake of simplicity, in this paper, in addiction to these assumptions, we will stick to the case when $\phi'(0)<\phi(0)$ (the other case, studied in \cite{mtns}, being a simpler and less interesting case than this). From this moment on we will refer to \emph{standard strong assumptions} (s.s.a.) on $\phi(z)$ to indicate $\phi\in C^2([0,1])$ such that, $\forall z\in[0,1]$,
\be\label{standard}\phi'(z)\geq 0\quad\text{,}\quad\phi''(z)\leq 0\quad\text{and}\quad\phi'(0)<\phi(0)\ee

As for the case of the SIS dynamics, if the parameters of the model are not degenerate, the pure configuration $0\1$ is the only absorbing state of $X(t)$ and from every configuration there is a non zero probability of reaching it in finite time. Consequently, with probability one, $X(t)=0\1$ for $t$ sufficiently large. Our aim is to study the behavior of the system before the absorbing event. Indeed, the maximum diffusion and the persistence of $1$'s in the population will be shown to exhibit a fundamental dependence on the parameter $\beta$, the function $\phi$, the topology of the graph, and, possibly, the initial fraction of $1$'s. 
 
The analysis will be carried on by considering the diffusion of $1$'s in the population whose dynamics is described by the process $Z(t)=z({X(t)})$, taking values in $\mc S_N=\{0,1/N,\dots , 1\}$. For a general graph, the process $Z(t)$ is not Markovian since the distribution of $1$'s in a neighborhood of a node is in general different from the global distribution of $1$'s in the population. 

To study $Z(t)$ on a general graph, we have to introduce the idea of \emph{active edges}. An edge $(v,w)$ is called active at time $t$ if $X_v(t)=0$ and $X_w(t)=1$. If we now denote by $\xi(t)=\xi(X(t))$ the fraction of active edges at time $t$ (namely $\xi(t)$ is the ratio between the number of active edges at time $t$ and $\abs{  E}$), we have that $Z(t)$ is Markovian when conditioned to $\xi(t)$. In particular, when conditioned to $\xi(t)=\xi$, $Z(t)$ is a birth and death Markov jump process and its transition rates from the state $z$ to $z+1/N$ and $z-1/N$ are, respectively, given by
\be\label{model}\left\{\ba{l}\lambda^{+}(z,\xi)=\abs{E}\bar{d}^{-1}\beta\xi \phi(z)=N\beta\xi \phi(z)\\
 \lambda^{-}(z,\xi)=N z.
 \ea\right.\ee
Of course, the difficulty is in the fact that the process $\xi(t)$ is not explicitly known. The next section is devoted to recall the main results in the mean field case from \cite{mtns}, which is essentially the only case when $\xi(t)$ is a deterministic function of $Z(t)$ and, thus, $Z(t)$ itself is Markovian. In spite of the apparent simplicity of this example, we will see that all possible interesting phenomena are already present in it. The general case will be taken up in Section \ref{section general}.

\section{The mean field case}\label{section mean field}
In this section we assume the graph to be complete and we will also assume that self-loops are present so that each unit belongs to its own neighborhood (this assumption is simply for the sake of simplicity of notation and has no effect in the large scale analysis). 
Under this assumptions, the fraction of active edges is a deterministic function of $Z(t)$, in particular we have that $\xi(t)=Z(t)(1-Z(t))$. This immediately implies that $Z(t)$ is a Markov birth and death process and its transition rates from (\ref{model}) read
\be\label{transitions mean field} \left\{\ba{ll}\lambda^{+}(z)=N\beta z(1-z) \phi(z)\\ \lambda^{-}(z)=N z.
\ea\right.\ee
For such processes a quite complete analysis is available and will be developed below. Notice that in \cite{mtns}, we obtined (\ref{transitions mean field}) directly from the analysis of the process $X(t)$ when $G$ is complete and not passing through the fraction of active edges $\xi(t)$. Of course, in this case, the local gossip interaction and the global influence are somehow mixed together but some key interesting phenomena can already be observed here.  

At first we can consider the hydrodynamic limit, sending $N\to \infty$. For finite time ranges $[0, T]$,  Kurtz's theorem \cite{Kurtz} insures that if $Z(0)$ converges a.s. to $z_0$, then the process $Z(t)$ converges uniformly on $[0, T]$  to the solution $z(t)$ of the following Cauchy problem:
\begin{equation}\label{ode}
\begin{cases}
z'(t)=F(z)\\
z(0)=z_0.
\end{cases}
\end{equation}
where 
\begin{equation}\label{rhs}F(z):=\lim\limits_{N\to +\infty}\frac{1}{N}(\lambda^+(z)-\lambda^-(z))=\beta z(1-z)\phi(z)-z.\end{equation}
More precisely, for every $T>0\;\exists\;C_T>0$, such that for $N$ sufficiently large, the following exponential decay holds:
\begin{equation}\label{Kurtz}\P\left(\sup\limits_{0\leq t\leq T}\left| Z(t)-z(t)\right|>\epsilon\right)\leq 4\exp\left(-C_TN\epsilon^2\right).\end{equation}

The analysis of (\ref{ode}) follows from the analysis of the zeros of $F(z)$
\begin{lemma}\label{lemma zeroes} Assume that $\phi$ satisfies s.s.a. (\ref{standard}) and let 
\begin{equation}\label{beta star}\beta^*=\left[\max_{z\in [0,1]} (1-z)\phi(z)\right]^{-1}.\end{equation} Then,
\begin{enumerate}
\item $\beta<\beta^{*}$, then $F$ has one zero $0$ and $F'(0)<0$;
\item $\beta^*<\beta<\phi(0)^{-1}$, then $F$ has three zeros $0<z_u(\beta)<z_s(\beta)$ and $F'(0)<0$, $F'(z_u(\beta))>0$, and $F'(z_s(\beta))<0$.
\item $\beta>\phi(0)^{-1}$, then $F$ has two zeros $0<z_s(\beta)$, $F'(0)>0$ and $F'(z_s(\beta))<0$.
\end{enumerate}
With the understanding that if $\phi(0)=0$, case 2. is unbounded above and case 3. never shows up. Points $z_s(\beta)$ and $z_u(\beta)$ can be characterized as follows
\be\label{equilibria} \begin{array}{rcl}z_u(\beta)&=&\min\{z>0\;|\;
\beta(1-z)\phi(z)-1=0\},\\\ z_s(\beta)&=&\min\{z>z_u(\beta)\;|\; \beta(1-z)\phi(z)-1=0\},\end{array}\ee
for those cases when they exist.
\end{lemma}

\begin{proof}
Let $f(z)=\beta(1-z)\phi(z)-1$. It is straightforward to check, under (\ref{standard}), that $f(z)$ is a concave function, that $f(1)<0$, that $f'(1)<0$ and that $f'(0)>0$. Hence the function $f(z)$ is increasing in a neighborhood of $0$ and decreasing in a neighborhood of $1$. Hence, $f(z)$ presents a global maximum point $z_{\rm max}\in (0,1)$ and $f(z_{\rm max})=\beta(\beta^*)^{-1}-1$. In case 1., $f(z)$ has no zeros. In case 2., instead, $f(z)=0$ has two distinct solutions $z_u(\beta)<z_s(\beta)$. Finally, in case 3., $f(z)$ has one zero $z_s(\beta)>0$. Sign of derivatives in $0$ can be checked directly, while other signs follow from the monotonicity properties of $f$.
\end{proof}
{\remark
In the special case $\phi(z)=z$, which will be used in the simulations of Section \ref{section specific}, we have that for $\beta\geq\beta^*=4$, explicit computation show that:
\begin{equation}\label{mean}
z_u(\beta)=\frac{1}{2}-\frac{1}{2}\sqrt{1-\frac{4}{\beta}},\quad\quad
z_s(\beta)=\frac{1}{2}+\frac{1}{2}\sqrt{1-\frac{4}{\beta}}.
\end{equation}}

Now the asymptotic behavior of (\ref{ode}) is an immediate consequence of previous result:

\begin{proposition}\label{mean_phi}
Assume $\phi$ satisfies s.s.a. (\ref{standard}). Then
\begin{enumerate}
\item if  $\beta<\beta^{*}$, then $z(t)\rightarrow 0$ $\forall z_0$;
\item if $\beta^*<\beta<\phi(0)^{-1}$, then $z(t)\rightarrow 0$ $\forall z_0< z_u(\beta)$ and $z(t)\rightarrow z_s(\beta)$ $\forall z_0> z_u(\beta)$;
\item if $\beta>\phi(0)^{-1}$, then $z(t)\rightarrow z_s(\beta)$ $\forall z_0\neq 0$.
\end{enumerate}
With the understanding that if $\phi(0)=0$, case 2. is unbounded above and case 3. never shows up.
\end{proposition}

Differently from the SIS dynamics, we witness here a regime $\beta^*<\beta<\phi(0)^{-1}$ where the asymptotics depends on the initial condition. Considering that Kurtz's theorem only guarantees convergence to the solution of the ODE for bounded range times, it is not a-priori clear what type of information on the original process $Z(t)$ can be gained when $N$ is large but finite. A key question is if the bifurcation phenomena described in Proposition \ref{mean_phi} (in particular the dependence on the initial condition) admit a sound interpretation at the level of the transient behavior of the process $Z(t)$. The answer is on the affirmative and relies on the analysis of the sojourn and absorbing times of the process $Z(t)$ expressed in the following result which essentially shows two facts.  When in the ODE there is a phenomenon of convergence to a stable non zero equilibrium, then the process $Z(t)$ with overwhelming probability\footnote{With overwhelming probability means that the probability of the event converges to $1$ exponentially fast as $N\to\infty$.} gets close to such point in finite time (due to \ref{Kurtz}) and then remains close to it for a time which is exponentially long in $N$ (as we will prove below). When instead in the ODE we have convergence to $0$, then, with overwhelming probability, in finite time the process gets close to $0$ and remains there ever since.

\begin{theorem}\label{absorbing}
Assume $\phi$ satisfies s.s.a. (\ref{standard}). Hence, for every $\eps >0$ we can find $C_\eps>$ and  $T_{\eps}>0$ for which the following holds true, if $N$ is sufficiently large,
\begin{enumerate}
\item if $\beta<\beta^{*}$, then $\forall z$, $$\P_z\left(\sup\limits_{t\geq T_\eps}Z(t)>\eps \right)<e^{-C_\eps N};$$
\item if $\beta^*<\beta<\phi(0)^{-1}$, then $\forall z<z_u(\beta)-\eps$,
$$\P_z\left(\sup\limits_{t\geq T_\eps}Z(t)>\eps \right)<e^{-C_\eps N},$$
and $\forall z>z_u(\beta)+\eps$,
$$\P_z\left(\inf\limits_{t\in [T_\eps, T_\eps +e^{C_\eps N}]}Z(t)<z_s(\beta)-\eps \right)<e^{-C_\eps N};$$
\item if $\beta>\phi(0)^{-1}$, then $\forall z>\eps$,
$$\P_z\left(\inf\limits_{t\in [T_\eps, T_\eps +e^{C_\eps N}]}Z(t)<z_s(\beta)-\eps \right)<e^{-C_\eps N}.$$
\end{enumerate}
\end{theorem}
The proof of Proposition \ref{absorbing} is based on a couple of technical lemmas. We start with this one dealing with the time needed by a death and birth process to proceed 'against' the mean drift.

\begin{lemma}\label{lemma time1}
Let $Z(t)$ be a  birth and death process on the state-space $\mc S_N$ with transitions rates, respectively, $\lambda^+(z)$ and $\lambda^-(z)$. Let  
$\mu:=\max_z[\lambda^+(z)+\lambda^-(z)]$.
Assume there exists an interval $(z_0-\eps, z_0+\eps)\subseteq (0,1)$, where $z_0\in (0,1)$ and $\eps >0$, such that
\begin{equation}\label{condition}
\lambda^{+}(z)\geq (1+\delta)\lambda^{-}(z),\quad\forall z\in  S_N\cap (z_0-\eps, z_0+\eps)
\end{equation}
for some $\delta >0$.
Then, for any $z>z_0$,
\be\label{obj lt1}\P_z\left(\exists t\in [0, \mu^{-1}e^{C\eps N}]\,|\, Z(t)<z_0-\eps \right)<10e^{-C\eps N}\ee
for a suitable constant $C>0$ only depending on $\delta$.
\end{lemma}

\begin{proof}
Let $\Lambda(t)$ be the number of jumps in the process $Z(t)$ in the time interval $[0,t]$. Clearly, $\Lambda(t)$ is dominated by a Poisson process of energy $\mu$. 
Let $\tilde Z(k)$ be the state of the process immediately after the $k$-th jump. Let $$A(t):=\{k=1,\dots ,\Lambda(t)\;|\;\tilde Z(k-1)\in (z_0-\eps, z_0+\eps)\}.$$
Let $\xi_k$ be the Bernoulli r.v. which is $1$ when the $k$-th jump in $Z(t)$ corresponds to a jump to the right.  Clearly, if $\tilde Z(k)\in (z_0-\eps, z_0+\eps)$, then,
\be\label{sojour-ber}\P(\xi_k=1)= \frac{\lambda^{+}(\tilde Z(k-1))}{\lambda^+(\tilde Z(k-1))+\lambda^-(\tilde Z(k-1))}\geq p=\frac{1+\delta}{2+\delta}>\frac{1}{2}.\ee
Fix a time range $T$ and notice that
$$\P(\Lambda(T)\geq 3\mu T)\leq\sum\limits_{k=\lceil 3\mu T\rceil}^{+\infty}e^{-\mu T}\frac{(\mu T)^k}{k!}\leq \frac{(\mu T)^{\lceil 3\mu T\rceil}}{\lceil 3\mu  T\rceil!}\leq \left(\frac{e}{3}\right)^{\lceil 3\mu T\rceil}_.$$
Conditioning on $\Lambda(T)$, we now estimate 
\be\label{sojour-estim1}\ba{l}\P_z\left(\inf\limits_{t\in [0, T]}Z(t)<z_0-\eps \right)=\\=\displaystyle\sum_s \P_z\left(\inf\limits_{t\in [0, T]}Z(t)<z_0-\eps \;|\; \Lambda(T)=s\right)\P(\Lambda(T)=s)=\\
\leq\displaystyle\sum_{s\leq 3\mu T}\P_z\left(\inf\limits_{t\in [0, T]}Z(t)<z_0-\eps \;|\; \Lambda(T)=s\right)\P(\Lambda(T)=s)+\left(\frac{e}{3}\right)^{\lceil 3\mu T\rceil}_.\ea\ee
As $\mu T$ will be exponential in $N$, this last term will play no significant role in our future estimations. We now concentrate on the summation term of the right hand side of (\ref{sojour-estim1}). The estimation we carry on is simply based on the fact that for the process $Z(t)$ to go below $z_0-\eps$ starting from above $z_0$, there must exist a sequence of  $l\geq\eps N$ consecutive jumps while the process is in $(z_0-\eps, z_0+\eps)$ for which the number of left transitions minus the number of right transitions is above $\eps N$. Henceforth, assuming that $z>z_0$ and using (\ref{sojour-ber}), we have that
\begin{equation}\label{sojour-estim2}\begin{array}{l}\P_z\left(\inf\limits_{t\in [0, T]}Z(t)<z_0-\eps \;|\; \Lambda(T)=s\right)\leq\\\leq
\displaystyle\sum\limits_{k=1}^s\sum\limits_{l=\lceil\eps N\rceil}^s\P\left(k+i\in A(t)\;\forall i=0,\dots ,l-1\,,\;\sum_{i=0}^{l-1}\xi_{k+i}\leq \frac{l}{2}-\delta \frac{N}{2}\right)\leq\\
\leq\displaystyle\sum\limits_{k=1}^s\sum\limits_{l=\lceil\eps N\rceil}^s\sum\limits_{h=0}^{l/2-\eps N/2}{l\choose h} p^h(1-p)^{l-h}
\leq \displaystyle\sum\limits_{k=1}^s\sum\limits_{l=\lceil\eps N\rceil}^se^{-Cl}
\leq s^2e^{-CN\eps},\end{array}\end{equation}
where the third inequality follows from the classical Chernoff bound. The constant $C$ only depends on $p$ and thus, ultimately, on $\delta$.
From (\ref{sojour-estim1}) and (\ref{sojour-estim2}), we obtain
$$\P_z\left(\inf\limits_{t\in [0, T]}Z(t)<z_0-\eps \right)\leq (3\mu T)^2e^{-CN\eps}+\left(\frac{e}{3}\right)^{\lceil 3\mu T\rceil}.$$
Finally, choosing $T=\mu^{-1}e^{\frac{C}{3}\eps N}$ and using the fact that $(e/3)^x<x^{-2}$ for all $x>0$, we immediately obtain the result.

\end{proof}

\begin{lemma}\label{lemma time2}
Let $Z(t)$ be a birth and death process on the state-space $ S_N$ with transitions rates, respectively,  $\lambda ^+(z)$ and $\lambda ^-(z)$. 
Assume that $\lambda ^+(0)=\lambda ^-(0)=0$ and that there exists $\eps >0$ such that
\begin{equation}\label{condition2}
\lambda ^{-}(z)\geq (1+\delta)\lambda ^{+}(z),\quad\forall z\in  S_N\cap (0,2\eps],
\end{equation}
for some $\delta >0$. Then, called $C=\ln (1+\delta)$, for any $z<\eps $,
\be\P_z\left(\exists t\geq 0\,|\, Z(t)> 2\eps \right)<\eps Ne^{-C\eps N}.\ee
\end{lemma}
\begin{proof}
Put, for $k=0,\dots , \lceil 2\delta N\rceil$, $e_k=\P_{k/N}\left(\exists t\geq 0\,|\, Z(t)\geq\lceil2\eps N\rceil/N  \right)$. A straightforward argument based on conditioning on the transition at the first jump gives
\be e_k=\frac{\lambda ^+(\frac{k}{N})}{\lambda ^+(\frac{k}{N})+\lambda ^-(\frac{k}{N})}e_{k+1}+\frac{\lambda ^-(\frac{k}{N})}{\lambda ^+(\frac{k}{N})+\lambda ^-(\frac{k}{N})}e_{k-1},\ee
which can be rewritten in a more compact way as
\be (e_{k+1}-e_k)=\frac{\lambda ^-(\frac{k}{N})}{\lambda ^+(\frac{k}{N})}(e_k-e_{k-1}).\ee
Using the boundary condition $e_0=0$, we immediately obtain that
\be\label{eq:recursive}(e_{k+1}-e_k)=\prod\limits_{j=1}^k\frac{\lambda ^-(\frac{j}{N})}{\lambda ^+(\frac{j}{N})}e_1.\ee
From (\ref{eq:recursive}) and the boundary condition $e_{\lceil\eps N\rceil}=1$ we get
\be\label{eq:recursive2}e_1=\left(\sum\limits_{k=0}^{\lceil 2\eps N\rceil -1}\prod\limits_{j=1}^k\frac{\lambda ^-(\frac{j}{N})}{\lambda ^+(\frac{j}{N})}\right)^{-1}\leq \left(\prod\limits_{j=1}^{\lceil 2\eps N\rceil -1}\frac{\lambda ^-(\frac{j}{N})}{\lambda ^+(\frac{j}{N})}\right)^{-1}_.\ee
Finally, from (\ref{eq:recursive}) and (\ref{eq:recursive2}) we obtain
$$e_{\lfloor \eps N\rfloor}=\sum\limits_{k=0}^{\lfloor \eps N\rfloor -1}\prod\limits_{j=1}^k\frac{\lambda ^-(\frac{j}{N})}{\lambda ^+(\frac{j}{N})}e_1
\leq \lfloor \eps N\rfloor\frac{\prod\limits_{j=1}^{\lfloor \eps N\rfloor-1}\frac{\lambda ^-(\frac{j}{N})}{\lambda ^+(\frac{j}{N})}}
{\prod\limits_{i=1}^{\lceil 2\eps N\rceil -1}\frac{\lambda ^-(\frac{i}{N})}{\lambda ^+(\frac{i}{N})}}\leq  \lfloor \eps N\rfloor (1+\delta)^{-\eps N}.$$
which yields the thesis.
\end{proof}
\begin{proof} {\it (of Theorem \ref{absorbing})}
As a general remark notice that the sign of the right-hand side of (\ref{ode}) is positive (negative) if and only if $\lambda^+(z)/\lambda^-(z)$ is, respectively, above (below) $1$.

Item 1 follows from the corresponding Item 1 of Lemma \ref{mean_phi} that, for any $\epsilon >0$, there exists $T_\eps>0$ such that $z(T_\eps)<\eps/4$ for any $z(0)=z$. 
We now estimate as follows:
$$\ba{rcl}\P_z\left(\sup\limits_{t\geq T_\eps}Z(t)>\eps \right)&\leq& \P_z\left(\sup\limits_{t\geq T_\eps}Z(t)>\eps\;|\; Z(T_\eps)< \frac{\eps}{2} \right)\\&+&\P_z(Z(T_\eps))\geq\frac{\eps}{2}).\ea$$
Conclusion now follows by estimating the first term using Lemma \ref{lemma time2} and the second one using (\ref{Kurtz}).

Item 3 follows from the corresponding Item 3 of Proposition \ref{mean_phi} that, for any $\epsilon >0$, there exists $T_\eps>0$ such that $z(T_\eps)>z_s-\eps/4$ if $z(0)\geq\eps$. We now estimate as follows:
$$\begin{array}{rcl}\P_z\left(\inf\limits_{t\geq T_\eps}Z(t)<z_s-\eps \right)&\leq& \P_z(Z(T_\eps))\leq z_s-\frac{\eps}{2})\\&+&\P_z\left(\inf\limits_{t\geq T_\eps}Z(t)<z_s-\eps\;|\; Z(T_\eps)> z_s-\frac{\eps}{2} \right).
\end{array}$$
Conclusion now follows by estimating the first term using using (\ref{Kurtz}) and the second one using Lemma \ref{lemma time1}.

Finally, Item 2 follows from Item 2 of Lemma \ref{mean_phi}, by similar arguments in dependence of the initial condition of the process.

\end{proof}

\section{Analysis on general graphs}\label{section general}
For the SIS model, estimation of the mean absorbing time has been carried on general graphs \cite{draief}, \cite{ahn}. In particular, fast extinction results have been obtained (see Theorem 8.2 in \cite{draief}) by upper bounding the original process with another one whose transition rates depend linearly on the state variable $\bf x$ and for which, consequently, the moment analysis turns out to be particularly simple. The key graph parameter in this estimation happens to be the spectral radius of the corresponding adjacency matrix. On the other hand, slow extinction has been analyzed by essentially estimating the fraction of active edges in terms of bottleneck ratios in the graph and then lower bounding with a simple birth and death process.

Following the techniques developed in \cite{draief}, in this section we will partially extend the results contained in Theorem \ref{absorbing} to more general sequences of large scale graphs. As we will see, the presence of the term $\phi(z)$ will pose a number of technical issues which are absent in the SIS model. In particular, in order to prove the existence of an intermediate regime where evolution depends on the initial condition, we will need to substantially extend the upper bound technique employed in \cite{draief} and carrying on a detailed second moment analysis of the bounding process.

Assume we have fixed a strongly connected graph, $ G=( V,  E)$. We denote with $A\in\{0,1\}^{V}$ the adjacency matrix of $ G$ ($A_{uv}=1$ iff $(u,v)\in  E$) and by $\rho_A$ its spectral radius. Consider a jump Markov process $X(t)$ evolving on $\{0,1\}^{V}$ having transition rates given by (\ref{transitions}). We recall that $Z(t)=z(X(t))$ denotes the total fraction of $1$'s in the population and $\xi(t)=\xi(X(t))$ the fraction of active edges. In the following three subsections we will provide a lower bound and two different upper bounds for the process, respectively.

\subsection{A bottleneck-based lower bound}
The following result allows to lower bound the process $Z(t)$ with a jump Markov birth and death process using an argument similar to the one used in \cite{como}. We first recall the notion of Cheeger constant (also called bottleneck ratio, introduced in \cite{cheeger}) of graph $ G=( V,  E)$:
\begin{equation}\label{bottleneck}
\gamma=\gamma_G=\inf_{{U}\subset {V}}\frac{\abs{\{(u,v)|u\in{U},\,v\in{{V}\setminus{U}}\}}}{\min\left\{\abs{{U}},\abs{{V}\setminus{U}}\right\}}.
\end{equation}

\begin{proposition}\label{lower bound}
There exists a coupling of the process $X(t)$ with a jump birth and death process $\tilde Z(t)$ over $\mc S_N$  having transition rates
\be\left\{\ba{l}\label{transitions lower}\tilde \lambda^+(z)=N\beta{\bar d}^{-1}\gamma z(1-z)\phi(z)\\ \tilde \lambda^-(z)=N z,\ea\right.\ee
in such a way that $Z(t) \geq \tilde Z(t)$ for all $t$.
\end{proposition}
\begin{proof} 
Let ${U}$ be the set of all agents with state equal to $0$. From (\ref{bottleneck}) we obtain
\[
\gamma\leq\frac{\xi\abs{{E}}}{\min\left\{z\abs{{V}},(1-z)\abs{{V}}\right\}}\leq\frac{\xi\abs{{E}}}{z(1-z)\abs{V}}=\frac{\xi \bar{d}\abs{{V}}}{z(1-z)\abs{{V}}}=\frac{\xi\bar{d}}{z(1-z)}.
\]
This implies that the number of active edges satisfies the inequality $\xi\geq{\gamma}\bar{d}^{-1} z(1-z)$. This yields, from (\ref{model}),
$$
\lambda^{+}(z,\xi)\geq \tilde \lambda^+(z)\quad\quad\text{and}\quad\quad\lambda^{-}(z,\xi)= \tilde \lambda^-(z).$$
It is now sufficient to apply a simple coupling argument analogous to the one used in the proof of Theorem 8.8 of \cite{draief}.
\end{proof}
\begin{corollary}\label{cor lower bound}
Put $z_u=z_u(\beta{\bar d}^{-1}\gamma)$ and $z_s=z_s(\beta{\bar d}^{-1}\gamma)$ as defined in (\ref{equilibria}). For every $\eps >0$ we can find $C_\eps>$ and  $T_{\eps}>0$ for which the following holds true, if $N$ is sufficiently large,
\begin{enumerate}
\item if ${\bar d}\gamma^{-1}\beta^*<\beta$, then, $\forall z>z_u+\eps$,
$$\P_z\left(\inf\limits_{t\in [T_\eps, T_\eps +e^{C_\eps N}]}Z(t)<z_s-\eps \right)<e^{-C_\eps N};$$
\item if, moreover, ${\bar d}\gamma^{-1}\phi(0)^{-1}<\beta$, then, $\forall z>\eps$,
$$\P_z\left(\inf\limits_{t\in [T_\eps, T_\eps +e^{C_\eps N}]}Z(t)<z_s-\eps \right)<e^{-C_\eps N}.$$
\end{enumerate}
\end{corollary}
\begin{proof} It follows from Proposition \ref{lower bound} that we can lower bound $Z(t)$ with the birth and death process $\tilde Z(t)$ having transition rates as in (\ref{transitions lower}). Confronting with (\ref{transitions mean field}), we deduce that $\tilde Z(t)$ coincide with the mean field model with $\beta$ replaced with $\beta{\bar d}^{-1}\gamma$ and the result then follows from items 2b. and 2c. of Theorem \ref{absorbing}.
\end{proof}

\subsection{A degree-based upper bound}
In this subsection we start with a simple upper bound which depends only on the degrees of the nodes in the graph. Let $\Delta$ be the maximum in-degree in $ G$, then the following proposition holds.
\begin{proposition}\label{upper degree} 
There exists a coupling of the process $X(t)$ with a jump birth and death process $\tilde Z(t)$ over $\mc S_N$  having transition rates
\be\left\{\ba{l}\label{transitions upper}\tilde\lambda^+(z)=\Delta \bar{d}^{-1} \beta z\phi(z)\\ \tilde \lambda^-(z)= z,\ea\right.\ee
in such a way that $Z(t) \leq \tilde Z(t)$ for all $t$.
\end{proposition}
\begin{proof} This simply follows from the estimation
\[
\xi\leq\frac{\abs{\{(u,v)|X_u=0,X_v=1\}}}{\abs{  E}}\leq\frac{\Delta z n}{\bar{d}n}=\Delta\bar{d}^{-1}z.
\]
\end{proof}
\begin{corollary}\label{cor upper degree}
For every $\eps >0$ we can find $C_\eps>$ and  $T_{\eps}>0$ for which the following holds true: if $\beta<\bar{d}\Delta^{-1}\phi(1)^{-1}$, then, if $N$ is sufficiently large, $$\P_z\left(\sup\limits_{t\geq T_\eps}Z(t)>\eps \right)<e^{-C_\eps N}\quad\forall z.$$
\end{corollary}
\begin{proof} It follows from Proposition \ref{upper degree} that we can upper bound $Z(t)$ with the birth and death process $\tilde Z(t)$ having transition rates as in (\ref{transitions upper}). Under the assumption $\beta<\bar{d}\Delta^{-1}\phi(1)^{-1}$, Lemma \ref{lemma time2} can be applied on (\ref{transitions upper}) in the style of the result contained in item 2b. of Theorem \ref{absorbing}.
\end{proof}

\subsection{An SIS-based upper bound}
Moreover, there is also an evident upper bound of our process in terms of a classical SIS model $X_{\rm SIS}$ having transition rates
\be\label{transitions SIS}\left\{\begin{array}{rcl}
\lambda^{\rm SIS}_{{\bf y}, {\bf y}+\delta_v}&=&\beta\bar{d}^{-1}(1-y_v)\sum\limits_{w\in\mc N_v}y_w\phi(1)\\
\lambda^{\rm SIS}_{{\bf y}, {\bf y}-\delta_v}&=&y_v,
\end{array}\right.
\ee
in the sense that we can find a coupling between the two processes under which $Z(t)\leq Z_{\rm SIS}(t)=z(X_{\rm SIS}(t))$ for all $t$. This is useful in those situations where the SIS model yields a fast extinction, namely when $\beta <{\bar d}\rho_A^{-1}\phi(1)^{-1}$. This is the case considered in Theorem 8.2 of \cite{draief}. On the other hand, in the case when ${\bar d}\rho_A^{-1}\phi(1)^{-1}<\beta <{\bar d}\rho_A^{-1}\phi(0)^{-1}$, we would expect that, similarly to the mean field case, a transition phase in terms of the initial condition should show up in the style of the result contained in item 2b. of Theorem \ref{absorbing}. It is clearly not possible to carry on such analysis simply in terms of the SIS model as this last model does not exhibit such phenomenon. In \cite{draief} the analysis of the fast extinction case for the SIS model was done through a further upper bound in terms of a jump process whose transition rates depend linearly on the configuration vector $\bf y$ and then carrying on a first moment analysis. This same idea turns out to be quite useful to analyze our process, but carrying on a significantly more complex analysis also involving the second moment.

We start by introducing the Markov jump process $Y(t)$ over $\Theta=\N^{V}$ having  transition rates
\be\label{transitions SIS linear 2}\left\{\begin{array}{rcl}
\bar\lambda_{{\bf y}, {\bf y}+\delta_v}&=&\mu\sum\limits_{w\in\mc N_v}y_w\\
\bar\lambda_{{\bf y}, {\bf y}-\delta_v}&=&y_v,
\end{array}\right.
\ee
where $\mu$ is a constant. 

Notice that our original process, taking values in $\{0,1\}^{V}$, can be trivially extended to $\Theta$ by simply putting $\lambda_{{\bf y}, {\bf y}+\delta_v}=0$ if ${\bf y}_v>0$ and using the same expression for 
$\lambda_{{\bf y}, {\bf y}-\delta_v}=y_v$. In the case when $\beta <{\bar d}\rho_A^{-1}\phi(1)^{-1}$, if we put $\mu=\beta\bar{d}^{-1}\phi(1)$, then it is clear
that $\bar\lambda_{{\bf y}, {\bf y}+\delta_v}\geq \lambda_{{\bf y}, {\bf y}+\delta_v}$ for all $\bf y$ and for all $v$. We can then consider any coupling between $X(t)$  and $Y(t)$ such that $X(0)=Y(0)$ and $X(t)\leq Y(t)$ (component wise) for all $t$. Clearly, it holds $Z(t)\leq Z_Y(t)=z(Y(t))$ for all $t$. In the sequel we will study the behavior of $Z_Y(t)$ in general; as we will see this will allow to get fast extinction result in the above case and we will also be able, through this, to analyze the transition phase phenomenon for ${\bar d}\rho_A^{-1}\phi(1)^{-1}<\beta <{\bar d}\rho_A^{-1}\phi(0)^{-1}$.

From the fact that
the distribution $p(t)\in[0,1]^{\Theta}$ of $Y(t)$ satisfy the forward Kolmogorov equation
$\dot{p}=-pL(\bar \lambda)$ where $L(\bar \lambda)$
 is the Laplacian of the process (i.e. $L(\bar \lambda)_{xy}=\sum_{y'}\bar\lambda_{xy'}-\bar\lambda_{xy}$),
it easily follows that the first moment $M^{(1)}(t)=\E(Y(t))$ satisfies the ODE
\be\label{mean ODE} \dot{M}^{(1)}=(\mu A- I)M^{(1)}.\ee
We can thus estimate
\be\label{estim mean} ||M^{(1)}(t)||\leq \exp ((\mu\rho_A-1)t) ||Y(0)||,\ee
where $\rho_A$ is the spectral radius of $A$. This yields
\be\label{decay mean}\E[Z_Y(t)] \leq n^{-1}n^{1/2}\exp ((\mu\rho_A-1)t)||X(0)||=\exp((\mu\rho_A-1)t) Z(0)^{1/2}.\ee
Notice that choosing $\mu=\beta\bar{d}^{-1}\phi(1)$, it holds $\mu\rho_A-1=(\beta\bar{d}^{-1}\rho_A\phi(1)-1)$ so that $\mu\rho_A-1<0$ when $\beta <{\bar d}\rho_A^{-1}\phi(1)^{-1}$. In this case we have an exponential decay to $0$ of $\E[Z(t)]$. This is not yet sufficient to obtain a generalization of the convergence result 2a. in Theorem \ref{absorbing}. For that and also for discussing the case ${\bar d}\rho_A^{-1}\phi(1)^{-1}<\beta <{\bar d}\rho_A^{-1}\phi(0)^{-1}$, we will carry on a second order analysis.

To this aim, put $M^{(2)}=\E(Y(t)Y(t)^*)$ and $\Omega=M^{(2)} -M^{(1)}M^{(1)*} $. The following result holds.

\begin{proposition} $\Omega$ satisfies the ODE
\be\label{covariance ODE}\dot{\Omega}=\mu(A\Omega+\Omega A) -2\Omega+\mu{\rm diag}(AM^{(1)}) +{\rm diag}(M^{(1)}),\ee
with $\Omega(0)=0$. 
\end{proposition}
\begin{proof}
Using the Kolmogorov equation it follows that
\be\label{derivative covariance}\begin{array}{rcl} \dot{M}^{(2)}&=&\sum\limits_{x\in\Theta}\dot{p}_x x x^*\\
&=& \sum\limits_{x\in\Theta}\mu\sum\limits_{v\in V}p_{x-\delta_v}(A(x-\delta_v))_vxx^*+\sum\limits_{x\in\Theta}\sum\limits_{v\in V}p_{x+\delta_v}(x_v+1)xx^*\\
&-&  \sum\limits_{x\in\Theta}\mu\sum\limits_{v\in V}p_{x}(Ax)_vxx^*- \sum\limits_{x\in\Theta}\sum\limits_{v\in V}p_{x}x_vxx^*.\\
\end{array}\ee
The first two terms of (\ref{derivative covariance}) can be rearranged as follows
\be\label{term1} \begin{array}{l}
 \sum\limits_{x\in\Theta}\mu\sum\limits_{v\in V}p_{x-\delta_v}(A(x-\delta_v))_vxx^*=\\
=\mu\sum\limits_{v\in V}\sum\limits_{x\in\Theta}p_{x-\delta_v}(A(x-\delta_v))_v(x-\delta_v)(x-\delta_v)^* +\\
+\mu\sum\limits_{v\in V}\sum\limits_{x\in\Theta}p_{x-\delta_v}(A(x-\delta_v))_v\delta_v(x-\delta_v)^*+\\
 +\mu\sum\limits_{v\in V}\sum\limits_{x\in\Theta}p_{x-\delta_v}(A(x-\delta_v))_v(x-\delta_v)(\delta_v)^*+\\
 +\mu\sum\limits_{v\in V}\sum\limits_{x\in\Theta}p_{x-\delta_v}(A(x-\delta_v))_v\delta_v\delta_v^*=\\
 =\mu\sum\limits_{v\in V}\sum\limits_{x\in\Theta}p_{x}(Ax)_vxx^* +\mu\sum\limits_{v\in V}\sum\limits_{x\in\Theta}p_{x}(Ax)_v\delta_vx^*+\\
+\mu\sum\limits_{v\in V}\sum\limits_{x\in\Theta}p_{x}(Ax)_vx\delta_v^*+\mu\sum\limits_{v\in V}\sum\limits_{x\in\Theta}p_{x}(Ax)_v\delta_v\delta_v^*\\
=\mu\sum\limits_{v\in V}\sum\limits_{x\in\Theta}p_{x}(Ax)_vxx^* +\mu(AM^{(2)}+M^{(2)} A) +\mu{\rm diag}(AM^{(1)});
 \end{array}
 \ee
and
\be\label{term2} \begin{array}{l}
 \sum\limits_{x\in\Theta}\sum\limits_{v\in V}p_{x+\delta_v}(x_v+1)xx^*=\\
 =\sum\limits_{v\in V}\sum\limits_{x\in\Theta}p_{x+\delta_v}(x_v+1)(x+\delta_v)(x+\delta_v)^*+\\
 -\sum\limits_{v\in V}\sum\limits_{x\in\Theta}p_{x+\delta_v}(x_v+1)\delta_v(x+\delta_v)^*+\\
 -\sum\limits_{v\in V}\sum\limits_{x\in\Theta}p_{x+\delta_v}(x_v+1)(x+\delta_v)\delta_v^* +\sum\limits_{v\in V}\sum\limits_{x\in\Theta}p_{x+\delta_v}(x_v+1)\delta_v\delta_v^*=\\
  =\sum\limits_{v\in V}\sum\limits_{x\in\Theta}p_{x}x_vxx^*-\sum\limits_{v\in V}\sum\limits_{x\in\Theta}p_{x}x_v\delta_vx^*+\\
 -\sum\limits_{v\in V}\sum\limits_{x\in\Theta}p_{x}x_vx\delta_v^* +\sum\limits_{v\in V}\sum\limits_{x\in\Theta}p_{x}x_v\delta_v\delta_v^*=\\
 =\sum\limits_{v\in V}\sum\limits_{x\in\Theta}p_{x}x_vxx^*-2M^{(2)}+{\rm diag}(M^{(1)}).
 \end{array}
  \ee

Substituting (\ref{term1}) and (\ref{term2}) inside (\ref{derivative covariance}), we finally obtain
$$\dot{M}^{(2)}=\mu(AM^{(2)}+M^{(2)} A) -2M^{(2)}+\mu{\rm diag}(AM^{(1)}) +{\rm diag}(M^{(1)}).$$
Thesis now immediately follows differentiating ${M}^{(1)}{M}^{(1)*}$ with (\ref{mean ODE}) and coupling with (\ref{covariance ODE}).
\end{proof}
We can now study ${\rm Var} (Z_Y(t))=N^{-2}\1^*\Omega\1$. The following result holds
\begin{proposition}
\be\label{estim var 2}{\rm Var} (Z_Y(t))\leq N^{-1/2}\displaystyle
\frac{\mu\rho_A+1}{-\mu\rho_A+1}\exp ((\mu\rho_A-1)t)Z_Y(0)^{1/2}.
\ee
\end{proposition}
\begin{proof}
Let $\mc S(V)$ be the set of symmetric matrices over $V$ and let $\mc L:\mc S(V)\to\mc S(V)$ be the linear operator given by $\mc L(M)=\mu (AM+MA) -2M$. Then, by (\ref{covariance ODE}), we can represent the centered second moment as
$$\Omega(t)=\int\limits_0^t \exp((t-s)\mc L)U(s)\,{\rm d}s$$
where
\be\label{U}U(t)=\mu{\rm diag}(AM^{(1)}(t)) +{\rm diag}(M^{(1)}(t)).\ee
Hence,
\be\label{int_var}{\rm Var} (Z_Y(t))=N^{-2}\1^*\int\limits_0^t \1^*[\exp((t-s)\mc L)U(s)]\,{\rm d}s\,\1.\ee
Using the straightforward relation in (\ref{int_var}),
$$\mc \exp (t\mc L)M=\exp (t(\mu A-I))M\exp (t(\mu A-I)),$$
we can estimate as follows:
\be\label{estim var 1}{\rm Var} (Z_Y(t))\leq N^{-2}N^{1/2}\displaystyle\int\limits_0^t \exp (2(t-s)(\mu \rho_A-1))||U(s)||\,{\rm d}s\, N^{1/2},\ee
where $||U(s)||$ is the induced $2$-norm of $U(s)$. From (\ref{U}) we can estimate this norm as
$$||U(s)||\leq\mu\max_v|\delta_v^*AM^{(1)}|+\max_v|\delta_v^*M^{(1)}|\leq (\mu\rho_A+1)\exp ((\mu\rho_A-1)s) ||Y(0)||.$$
Inserting this estimation in (\ref{estim var 1}), we obtain the thesis.
\end{proof}

We are now ready to analyze the convergence behavior of the process $Z_Y(t)$ in the case when $\mu\rho_A<1$.

\begin{proposition}\label{prop conc ZY} Assume that $\mu\rho_A<1$.
\begin{enumerate} 
\item For every $\delta >0$ there exists a constant $C'_{\delta}>0$ such that, if $Z_Y(0)\leq a^2$, it holds
\be\label{moment estim 1}\P(\exists t\geq 0\,|\, Z_Y(t)>a+\delta)\leq C_{\delta}N^{-1/2}.\ee
\item For every $\epsilon >0$ there exists a time $T_{\epsilon}>0$ and a constant $C''_{\epsilon}>0$such that, for every $Z_Y(0)$,
\be\label{moment estim 2}\P(\exists t\geq T_{\epsilon}\,|\, Z_Y(t)>\epsilon)\leq C_{\epsilon}N^{-1/2}.\ee
\end{enumerate}
\end{proposition}

\begin{proof}
Consider the underlying discrete time MC $\tilde Y(k)$ for $k=0,1,\dots$ and the corresponding $Z_{\tilde Y}(k)=z(\tilde Y(k))$. 
The Poisson process $\Lambda(t)$ governing the jumps of $Y(t)$ has intensity $\nu=(\beta+1)N$ and it holds
\be\label{Poisson 1}\ba{rcl}{\rm Var}(Z_Y(t))&=&\sum\limits_{k=0}^{+\infty}{\rm Var}(Z_{\tilde Y}(k))\P(\Lambda(t)=k)\geq\\&\geq& {\rm Var}(Z_{\tilde Y}(\lfloor \nu t\rfloor))\P(\Lambda(t)=\lfloor \nu t\rfloor).\ea\ee
We can now estimate, using Stirling,
\be\label{Poisson 2} \P(\Lambda(t)=\lfloor \nu t\rfloor)=\frac{(\nu t)^{\lfloor \nu t\rfloor}}{\lfloor \nu t\rfloor !}e^{-\nu t}\geq \frac{(\nu t)^{\lfloor \nu t\rfloor}}{\lfloor \nu t\rfloor ^{\lfloor \nu t\rfloor}}\frac{e^{\lfloor \nu t\rfloor}}{\sqrt{2\pi \lfloor \nu t\rfloor}}e^{-\nu t}\geq \frac{1}{9\lfloor \nu t\rfloor}\ee
From (\ref{Poisson 1}) and (\ref{Poisson 2}) we obtain that, 
\be\label{variance discrete} {\rm Var}(Z_{\tilde Y}(k))\leq 9k{\rm Var}(Z_Y(k/\nu)),\quad\forall k=0,1,\dots \ee
By the way the initial condition has been chosen and by (\ref{decay mean}), we have that 

$$\begin{array}{rcl}\P(\exists t\geq 0\,|\, Z_Y(t)>a+\delta)&\leq&\P(\exists k\geq 0\,|\, Z_{\tilde Y}(k\lfloor \frac{\delta N}{2}\rfloor)>a+\frac{\delta}{2})\leq\\[8pt]
&\leq& \sum\limits_{k\geq 0}\P(|Z_Y(k\lfloor  \frac{\delta N}{2}\rfloor)-\E(Z_Y(k\lfloor \frac{\delta N}{2}\rfloor))|\geq \frac{\delta}{2})\leq\\[8pt]
&\leq&\displaystyle \frac{4}{\delta^2}\sum\limits_{k\geq 0}{\rm Var} (Z_{\tilde Y}(k\lfloor\frac{\delta N}{2}\rfloor)).\end{array}
$$
Inserting now the estimation (\ref{variance discrete}) and (\ref{estim var 2}), we immediately obtain item 1.

For item 2., it is sufficient to notice that, from (\ref{decay mean}) there exists a $T_{\epsilon} >0$ such that $\E(Z_Y(t))\leq \epsilon /2$ for all $t\geq T_{\epsilon}$ and then use again the variance estimation in a similar fashion.

\end{proof}

We are now ready to analyze our original process $Z(t)$.

\begin{corollary}\label{cor upper bound} 
For every $\epsilon >0$ $\exists\, C_{\epsilon}>0$ and $\exists\,T_{\epsilon}>0$ such that
\begin{enumerate}
\item if $\beta <{\bar d}\rho_A^{-1}\phi(1)^{-1}$, then
\be\label{concentration ZX 1}\P(\exists t\geq T_{\epsilon}\,|\, Z (t)>\epsilon)\leq C_{\epsilon}N^{-1/2};\ee
\item if ${\bar d}\rho_A^{-1}\phi(1)^{-1}<\beta <{\bar d}\rho_A^{-1}\phi(0)^{-1}$ and if $Z (0)\leq (z^*-2\epsilon)^2$, it holds
\be\label{concentration ZX}\P(\exists t\geq T_{\epsilon}\,|\, Z (t)>\epsilon)\leq C_{\epsilon}N^{-1/2},\ee
\end{enumerate}
where $z^*$ is the unique solution of $\phi(z^*)=\beta^{-1}\bar d\rho_A^{-1}$.
\end{corollary}
\begin{proof} 1. Straightforward consequence of the fact that $Z(t)\leq Z_Y(t)$ and of item 2. of Proposition \ref{prop conc ZY} applied with $\mu=\beta\bar{d}^{-1}\phi(1)$.

2. Consider the jump Markov process $Y(t)$ with transition rates given by (\ref{transitions SIS linear 2}) and $\mu=\beta \bar d^{-1}\phi(z^*-\epsilon)$. Notice that we still have
$$\mu\rho_A=\beta \bar d^{-1}\phi(z^*-\epsilon)\rho_A=\frac{\phi(z^*-\epsilon)}{\phi(z^*)}<1.$$
Notice that $\bar\lambda_{\bf y, \bf y+ \delta_v}\geq  \lambda_{\bf y, \bf y+ \delta_v}$ as long as $\bf y$ is such that $z(\bf y)\leq z^*-\eps$. Put
$$\bar T=\inf\{t\,|\, Y(t)>z^*-\eps\}.$$
We can establish a coupling between $X(t)$ and $Y(t)$ such that $X(0)=Y(0)$ and $X(t)\leq Y(t)$ for all $t< \bar T$.
Choose now $T_{\epsilon}$ so to satisfy item 2. of Proposition \ref{prop conc ZY}. It holds
$$\begin{array}{rcl}\P(\exists t\geq T_{\epsilon}\,|\, Z (t)>\epsilon)&=&\P(\exists t\geq T_{\epsilon}\,|\, Z (t)>\epsilon\,, \bar T=+\infty)+\\
&+&\P(\exists t\geq T_{\epsilon}\,|\, Z (t)>\epsilon\;, \bar T<+\infty)\leq\\
&\leq &\P(\exists t\geq T_{\epsilon}\,|\, Z_Y(t)>\epsilon)
+\P(\exists t\geq 0\,|\, Z_Y(t)>z^*-\epsilon).\end{array}$$
Result now follows by applying item 1. (with $a=z^*-2\eps$ and $\delta=\eps$) and item 2. of Proposition \ref{prop conc ZY} 
\end{proof}
\subsection{The core result}
The mail result of this paper is finally obtained by combining the three Corollaries \ref{cor lower bound}, \ref{cor upper degree} and \ref{cor upper bound} and recalling that the following inequalities always hold true:
\be\label{inequalities}\gamma\leq\rho_A\leq\bar d\leq\Delta\quad\quad\text{and}\quad\quad
\phi(1)^{-1}\leq\beta^*\leq\phi(0)^{-1},\ee 
since the first one and the third one are trivial, the second one is proven in \cite{collatz} and the last two are consequences of the monotonicity of $\phi(z)$.

\begin{theorem}\label{absorbing general}
Assume s.s.a. (\ref{standard}) hold on $\phi$. 
Assume we have fixed a graph $G$ having average degree $\bar d$, maximum degree $\Delta$, Cheeger constant $\gamma$ and spectral radius of the adjacency matrix $\rho_A$.
Let $z_u'\leq z_u''< z_s$ be points in $[0,1]$ defined by 
\be\label{points}\phi(\sqrt{z_u'})=\beta^{-1}\bar d\rho_A^{-1},\quad z_u''=z_u(\Delta \bar{d}^{-1} \beta),\quad z_s=z_s(\Delta \bar{d}^{-1}\beta),\ee
where $z_u(\cdot)$ and $z_s(\cdot)$ have been defined in (\ref{equilibria}). Depending on the conditions of the various parameters  each of this point may exist or not. Below, whenever we write them, we are implicitly affirming their existence.

For every $\eps >0$ we can find $C_\eps>$ and  $T_{\eps}>0$ for which the following holds true, if $N$ is sufficiently large,
\begin{enumerate}
\item if $\beta<\bar{d}\Delta^{-1}\phi(1)^{-1}$, then $\forall z$, $$\P_z\left(\sup\limits_{t\geq T_\eps}Z(t)>\eps \right)<e^{-C_\eps N};$$

\item if $ {\bar d}\gamma^{-1}\beta^*     <\beta <{\bar d}\rho_A^{-1}\phi(0)^{-1}$,  then $\forall z<z_u'-4\eps$,

$$\P_z(\exists t\geq T_{\epsilon}\,|\, Z (t)>\epsilon)\leq C_{\epsilon}N^{-1/2},$$
and $\forall z>z_{u}''+\eps$, 
$$\P_z\left(\inf\limits_{t\in [T_\eps, T_\eps +e^{C_\eps N}]}Z(t)<z_s-\eps \right)<e^{-C_\eps N};$$
\item if, ${\bar d}\gamma^{-1}\phi(0)^{-1}<\beta$, then $\forall z>\eps$,
$$\P_z\left(\inf\limits_{t\in [T_\eps, T_\eps +e^{C_\eps N}]}Z(t)<z_s-\eps \right)<e^{-C_\eps N}.$$
\end{enumerate}
\end{theorem}

{\remark\label{relaxed} Item 1 of Theorem \ref{absorbing general} can be relaxed using item 1 of Corollary \ref{cor upper bound}. In this case it reads
\begin{enumerate}
\item \emph{(bis) if} $\beta<\bar{d}\rho_A^{-1}\phi(1)^{-1}$\emph{, then }$\forall z$\emph{, }$$\P_z\left(\sup\limits_{t\geq T_\eps}Z(t)>\eps \right)<C_\eps N^{-1/2}.$$
\end{enumerate}
On the one hand this improves the bound on $\beta$ making it scaling with the inverse of the spectral radius instead of the maximum degree ($\rho_A\leq\Delta$). On the other hand, under this regime, the result becomes less strong in probability: Corollary \ref{cor upper bound} provides an upper-bound on the process $Z(t)$ only w.h.p.\footnote{A family of events $E_N$ is said to hold asymptotically almost surely (a.a.s.) if $\P[E_N]=1-o(1)$, with high probability (w.h.p.) if $\P[E_N]=1-o(N^{-c})$, for some $c>0$ and with overwhelming probability (w.o.p.) if $\P[E_N]=1-o(e^{-N})$, as $N\to\infty$} instead Theorem \ref{absorbing general} ensures it w.o.p.$^{1}$ Therefore, depending on the situation, it could be better to use one result or the other one. For example in deterministic regular graphs $\rho_A=\Delta$, therefore Theorem \ref{absorbing general} gives exactly the same bound of Corollary \ref{cor upper bound} with a stronger result in probability, thus the original result of Theorem \ref{absorbing general} is preferable. On the contrary in many random graphs, the graph parameters may converge to some constants (or functions of $N$) w.h.p., therefore any result w.o.p. is in any case lost. In these cases it is better to use the relaxed bound in Remark (\ref{relaxed}) since it provides a better threshold, with the same strength in probability.}

\section{Analytical and numerical results on specific topologies}\label{section specific} 
In this section we discuss the application of Theorem \ref{absorbing general} to specific sequences of graphs with increasing order. For the sake of simplicity we will always stick, in this section, to the case when $\phi(z)=z$ so that $\beta^*=4$ and only the two cases 1. and 2. in the theorem can possibly show up. 

We show that for sequences of graphs which we call regularly expansive (and which include important random graphs examples like Erd\H{o}s-R\'{e}nyi graphs and random configuration models) the theorem guarantees  the existence of a phase transition with respect to the parameter $\beta$ (which is sharp but from a multiplicative constant) between a situation where fast extinction always occurs and one where both fast extinction and long permanence may show up in dependence of the initial condition. 

We then present a number of numerical simulations\footnote{Realized with \textsc{MATLAB} with a sample size of 500 simulations for each level.} on these examples and also on other examples (Bar{\'a}basi-Albert graphs and grids) for which Theorem \ref{absorbing general} does not give any information. In these simulations, a sample path is called successful if it is not absorbed after a time $T=100$.
In the performed simulations, we notice that all successful sample paths, exhibit a fraction $z(T)$ of agents with the asset at time $T$ greater than $0.5$. This remark, on the one hand enforces the hypothesis of the presence of a phase transitions between a regime where the process is absorbed and another regime where the asset diffuses deeply in the population, on the other hand it makes our definition of success more consistent.

We recall below the graph parameters which need to be computed (or at least estimated) in Theorem \ref{absorbing general}:
\begin{itemize}
\item $\bar d$ and $\Delta$ are, respectively, the average and the largest degree of the graph;
\item $\gamma$ is the bottleneck ratio, $\rho_A$ is the spectral radius of the adjacency matrix.
\end{itemize}
and that $\gamma\leq\bar d\leq \rho_A \leq \Delta$, from inequalities (\ref{inequalities}).
A sequence of graphs $ G_N$ with increasing number $N$ of vertices is called $(a,e_1,e_2)$-regularly expansive if, for every $N$,
$$\bar d\Delta^{-1}\geq a\quad\quad\text{and}\quad\quad e_1\leq  \bar d\rho_A^{-1}\leq \bar d\gamma^{-1}\leq e_2.$$
Notice that, because of (\ref{inequalities}), we can always choose $e_1\geq a$.

Theorem \ref{absorbing general} can be reformulated for such families of graphs as follows:

\begin{corollary}\label{expansive}
Assume that $\phi(z)=z$ and that $ G_N$ is a $(a,e_1,e_2)$-regularly expansive sequence of graphs.
Let $z_u'\leq z_u''< z_s$ be points in $[0,1]$ defined by 
\be\label{points 2}z_u'=\beta^{-2}e_1^2,\quad z_u''=\frac{1}{2}-\frac{1}{2}\sqrt{1-\frac{4e_2}{\beta}},\quad z_s=\frac{1}{2}+\frac{1}{2}\sqrt{1-\frac{4e_1}{\beta}}.\ee
For every $\eps >0$ we can find $C_\eps>$ and  $T_{\eps}>0$ for which the following holds true, if $N$ is sufficiently large,
\begin{enumerate}
\item if $\beta<a$, then $\forall z$, $$\P_z\left(\sup\limits_{t\geq T_\eps}Z(t)>\eps \right)<e^{-C_\eps N};$$

\item if $\beta>4 e_2$,  then $\forall z<z_u'-4\eps$, $$\P_z(\exists t\geq T_{\epsilon}\,|\, Z (t)>\epsilon)\leq C_{\epsilon}N^{-1/2};$$
and $\forall z>z_{u}''+\eps$,
$$\P_z\left(\inf\limits_{t\in [T_\eps, T_\eps +e^{C_\eps N}]}Z(t)<z_s-\eps \right)<e^{-C_\eps N}.$$
\end{enumerate}
\end{corollary}
\begin{proof}
It is an immediate consequence of Theorem \ref{absorbing general}, of the explicit formulas (\ref{mean}), and the inequalities (\ref{inequalities}).
\end{proof}

Notice that, in order for being sure that there exists a transition with respect to the parameter $\beta$ from regime 1. to regime 2. we must have $a>0$. Instead the condition $e_1>0$ insures that there is the transition depending on the initial condition in regime 2.

Fundamental examples of regularly expansive graphs are the random Erd\H{o}s-R\'{e}nyi graphs and the random configuration model with fixed and bounded degree distribution. Below we recall the exact definition of the two ensembles and we discuss the application of Corollary \ref{expansive} to them.

\subsection{Erd\H{o}s-R\'{e}nyi graph}

The Erd\H{o}s-R\'{e}nyi graph $G(N,p)$ is the first random graph model introduced in 1959 in \cite{er}. $G(N,p)$ has $N$ nodes and each edge $(u,v)$ is present with a probability $p\in(0,1)$, independent on the other edges. Therefore the degree of each node is distributed as a binomial random variable with parameters $N-1$ and $p$, which means that the expected average degree is $(N-1)p$. Standard concentration results \cite{draief} show that, w.h.p., as $N\to\infty$,
$$\bar d\sim\Delta\sim Np.$$
We recall that w.h.p. means that i.e. the probability of not having the concentration above converges asymptotically to $0$ as some (negative) power of $N$, as $N\to\infty$. Here we will restrict to the case where $\frac{\ln{N}}{Np}\to 0$. Under this regime, it is proved in \cite{draief} (the first two results) and in \cite{chung} (the last one) that w.h.p. the graph is connected, and that 
$$\bar{d}\gamma^{-1}=2+o(1)\quad\quad\text{and}\quad\quad \bar{d}\rho_A^{-1}=1+o(1).$$ 
We can thus say that w.h.p. $G(N,p)$ is $(1-\delta, 1-\delta , 2+\delta)$-regularly expansive for any $\delta >0$ if $N$ is sufficiently large.

\begin{figure}[t]
\centering
\subfloat[Success depending on $\beta$\label{exp beta} ($z_0=1$).]{\includegraphics[scale=.5]{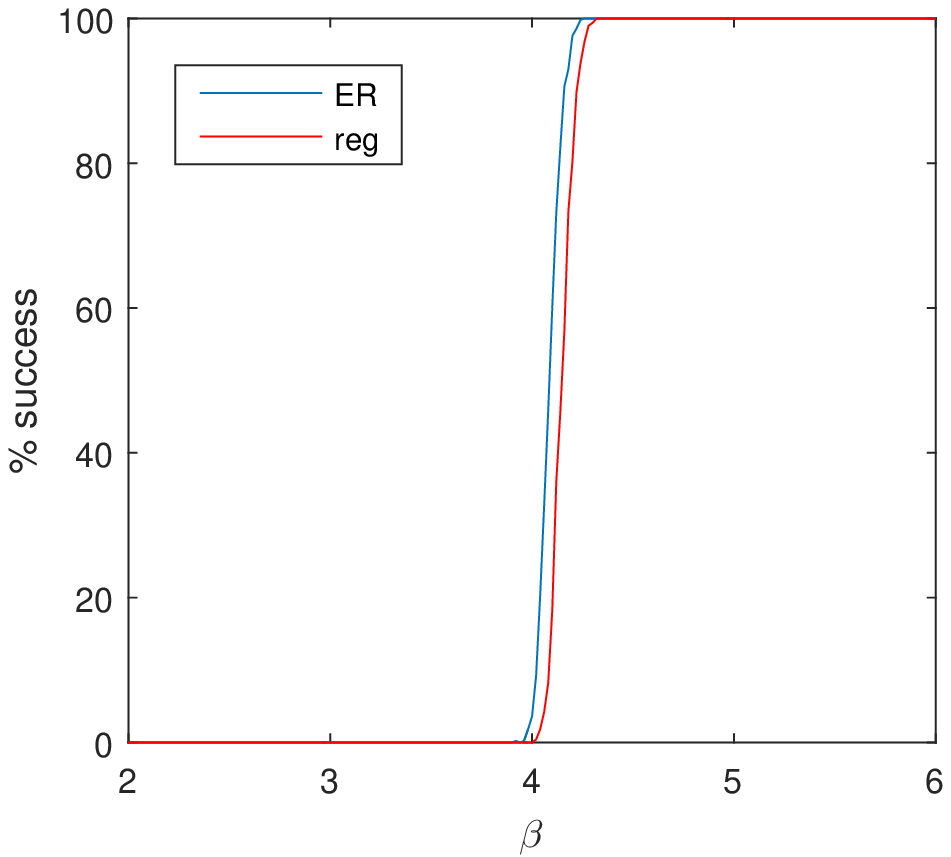}}\quad\subfloat[Success depending on $z_0$\label{exp z0} ($\beta=10$).]{\includegraphics[scale=.5]{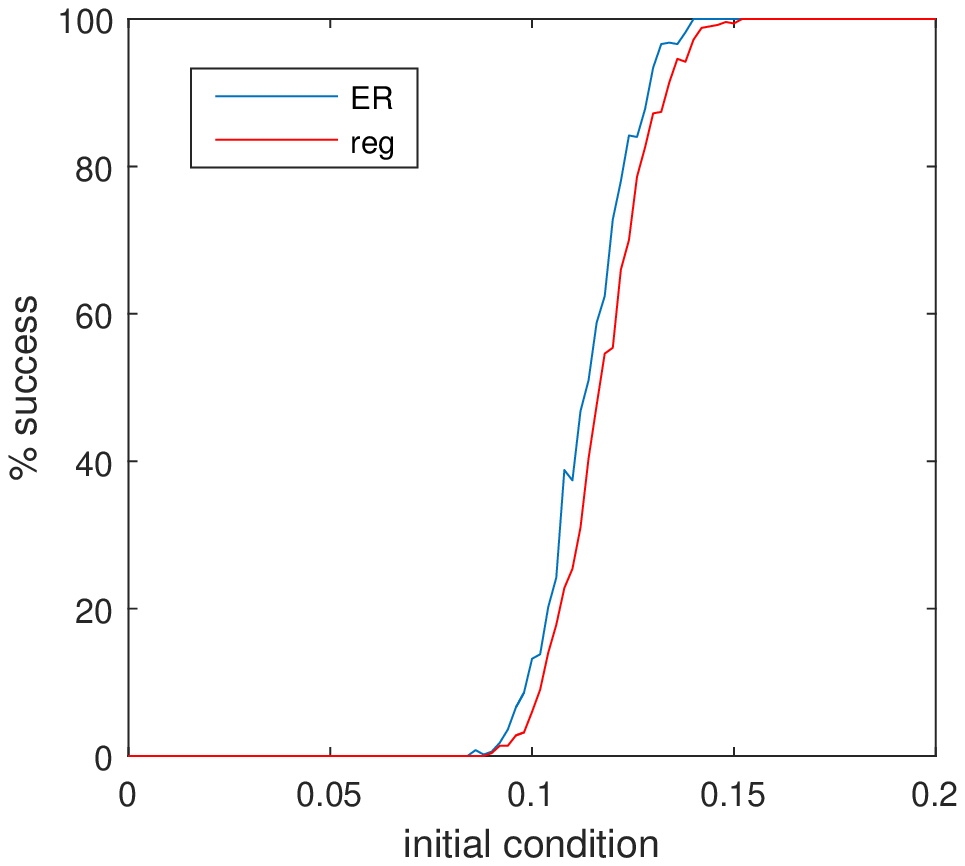}}
\caption{Simulations on random graphs with $N=1000$. In ER $p=0.05$, in the regular configuration model $\bar d=20$.}
\label{expander}
\end{figure}

Therefore the phase transition depending on $\beta$ in the case of connected ER graphs is sharp (but from a multiplicative constant). From the numerical simulations in Fig. \ref{exp beta}, it seems that the multiplicative constant actually shrinks to $1$, so that the phase transition is sharp. Observing Fig.  \ref{er2} we can also notice that the threshold in the initial condition seems to become sharper as $N\to \infty$, even if the two bounds $z_u'$ and $z_u''$  from Corollary \ref{expansive} does not converge to the same point. Therefore it looks like that the behavior of the model on ER graphs is really close to the results of mean field analysis in Theorem \ref{absorbing} with a sharp transition between the phases described in items $1$ and $2$ close to $\beta^*$ and a sharp threshold (as $N\to \infty$) for the initial condition $z_t$ in the phase described in item $2$.
\begin{figure}[h!]
\centering
\includegraphics[height=6cm]{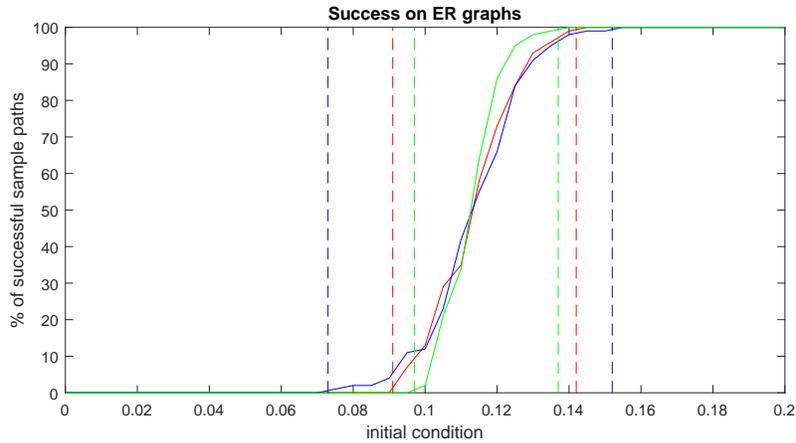}
\caption{Simulations on ER graphs with $\beta=10$ and $n=800$ (blue), $n=1200$ (red) and $n=1600$ (green). The vertical dotted lines are the estimated thresholds $z_u'$ and $z_u''$. As $N$ increases the transition seems to be sharper. Notice that the analytical thresholds from Corollary \ref{expansive} are $z_u'=0.01$ and $z_u''\simeq0.2764$.}\label{er2}
\end{figure}

\subsection{Configuration model}

Consider a probability distribution $q_d$ over the positive integers with $q_d=0$ if $d\leq 2$ and if $d> d_{max}$ is sufficiently large. The configuration model $G(N,{\bf q})$ is the sequence of symmetric random graphs with $N$ nodes whose degrees are independent  r.v.'s distributed according to $q_d$ and where connection is established through a random permutation (see \cite{durret} for details). Notice that, by construction, $3\leq\bar d\leq d_{max}$ and $\Delta\leq d_{max}$.

Moreover, from \cite{durret}, $\exists \alpha$ such that $\gamma\geq\alpha>0$ for all finite $N$ and w.h.p. as $N\to\infty$. Consequently, w.h.p.  $G(N,{\bf q})$ is $(3/d_{max}, 3/d_{max}, d_{max}/\alpha)$-regularly expansive and Corollary \ref{expansive} can be applied. Moreover, given a distribution $\bf q$, one can obtain more refined estimations on $\bar d$ and $\rho_A$ using probabilistic tools and taking advantage of the results from \cite{chung} and \cite{kumar}. Finally, even in this case, through numerical simulations it is possible to conjecture a behavior similar to ER graphs, as one can see in Fig. \ref{exp beta} (as the phase transition depending on $\beta$ is considered) and \ref{exp z0} (for the thresholds on the initial condition $z'_u$ and $z''_u$).

\subsection{Analysis on other graphs}
In the following section we consider some examples of non-expansive graph topologies where Corollary \ref{expansive} can not infer the presence of a phase transition. Nevertheless we will show, through numerical simulations, that such phenomena (or at least some of them) do take place.

The Bar{\'a}basi-Albert graph is a random graph model introduced in 1999 to represent social networks. Starting from an initial connected graph, each time a node is added and it is connected to $m$ existing nodes with a probability proportional to their degrees, until there are $N$ nodes (to see a precise definition of this model see \cite{ab}). 

\begin{figure}[h!]
\centering
\subfloat[Success depending on $\beta$\label{ab beta} ($z_0=1$).]{\includegraphics[scale=.50]{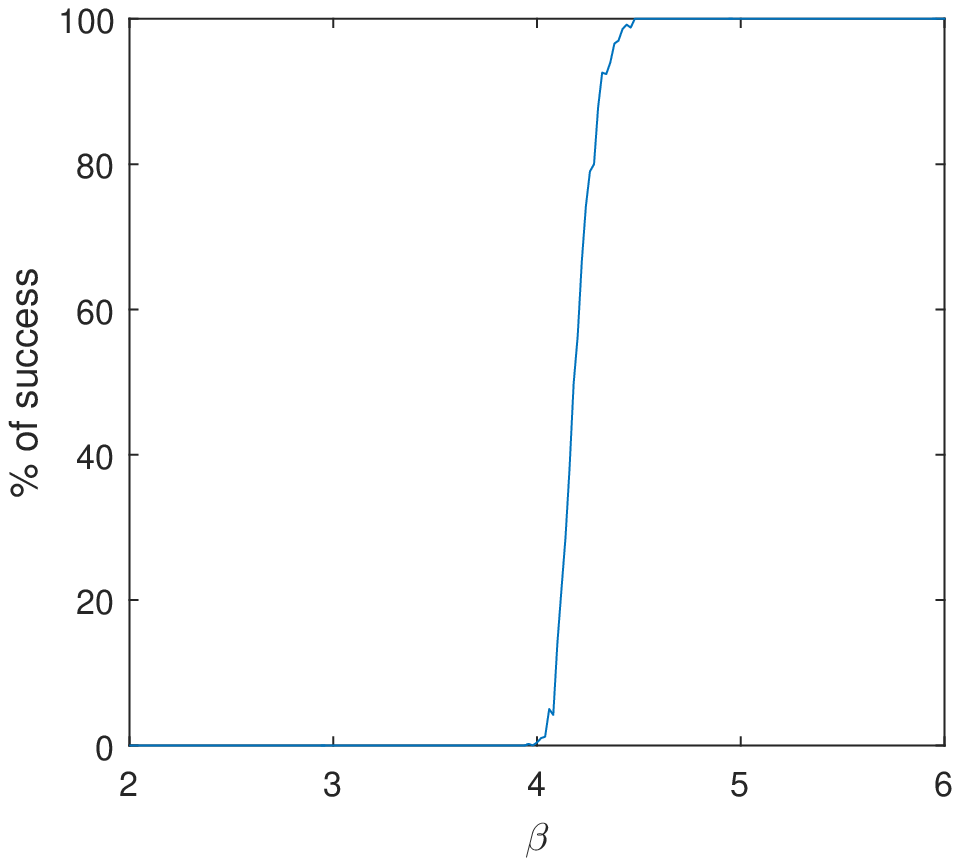}}\quad\subfloat[Success depending on $z_0$\label{ab z0} ($\beta=1$).]{\includegraphics[scale=.50]{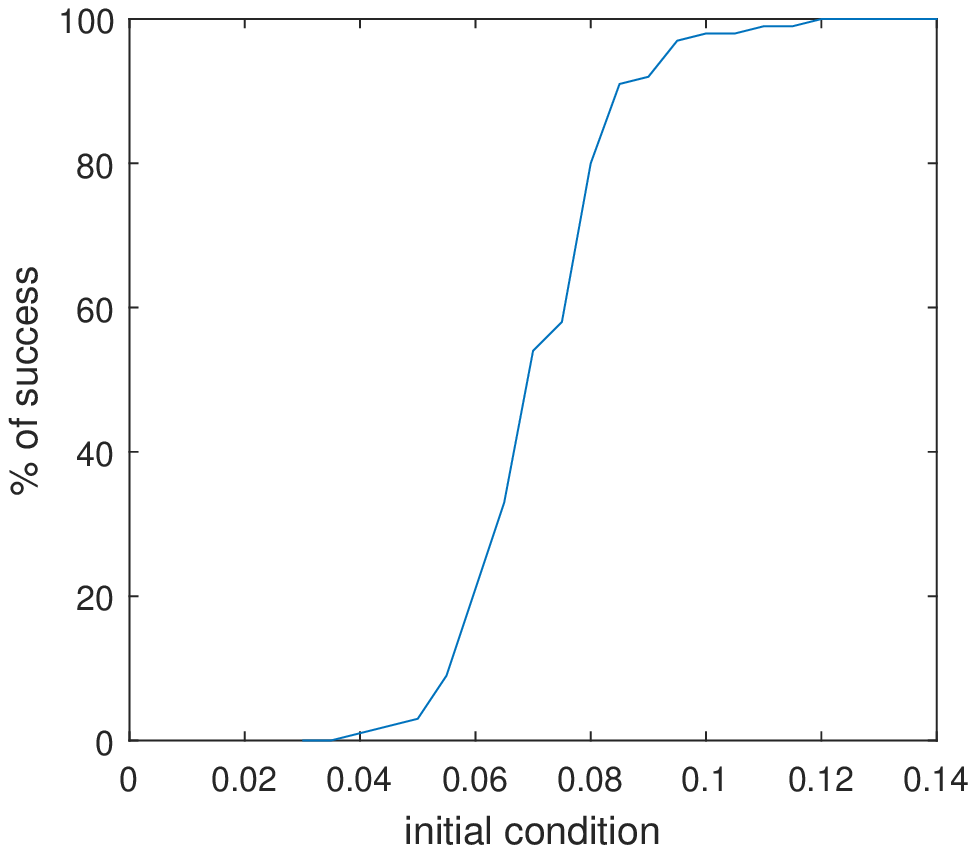}}
\caption{Simulations of the dynamics on AB graph with $N=1000$.}
\label{ab}
\end{figure}

This algorithm gives back a graph whose degree distribution is proved to follow asymptotically a power-law \cite{ab} (in particular $\P[d_v=k]\propto k^{-3}$). As $N\to\infty$ it is immediate to verify that $\bar{d}=m+o(1)$ (due to construction). On the other hand it is proved in \cite{chung} that w.h.p.
$$\Delta=n^{1/2}(1+o(1)),\quad \rho_A=n^{1/4}(1+o(1)),\quad 0<\alpha\leq \gamma=O(1).$$
Therefore, Bar{\'a}basi-Albert is only $(0, 0, m/\alpha +\delta)$-regularly expansive and Corollary \ref{expansive} can not infer any clue on transitions neither on $\beta$, nor on the initial condition.

Nevertheless, from Fig. \ref{ab beta} and \ref{ab z0} we can appreciate that it seems to exists a phase where success and failure are both possible depending on the initial condition. However, from our simulations, transition between those two behaviors seems to be smooth, even increasing $N$ (see Fig. \ref{ab z0}).

Another class of graphs on which Theorem \ref{absorbing general} and its Corollary do not provide any analytical proof of the existence of phase transitions is the class of the $k$-dimensional toroidal grids. We recall that a $1$-torus is a cyclic graph $C_n$ and a $k$-torus can be defined as the cartesian product between $k$ $1$-tori with $n^{1/k}$ nodes each \cite{vizing}. Therefore, given a dimension $k$, $\gamma\sim c N^{-k/2}$ for some constant $c$, while $\bar d=2k$ is constant, so that $\bar d/\gamma$ always diverges (a similar behavior occurs in hypercubes, where instead $\gamma= 1$ and $\bar d=\ln N$). Nevertheless, even in this case, simulations (see Fig. \ref{tori}) show the existence of an intermediate phase where the behavior of the model presents a transition depending on the initial condition. In this case the transition seems to be sharp, as one can appreciate in Fig. \ref{tori z0}.

\begin{figure}[h]
\centering
\subfloat[Success depending on $\beta$\label{tori beta} ($z_0=1$).]{\includegraphics[scale=.50]{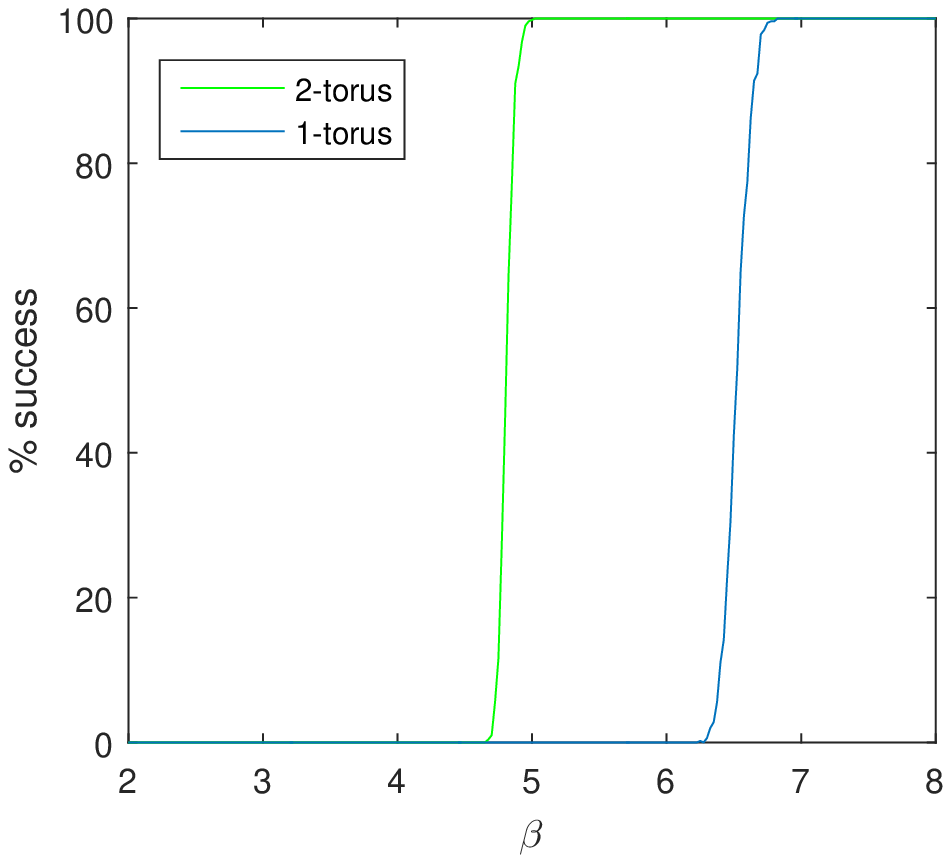}}\quad\subfloat[Success depending on $z_0$\label{tori z0} ($\beta=10$).]{\includegraphics[scale=.50]{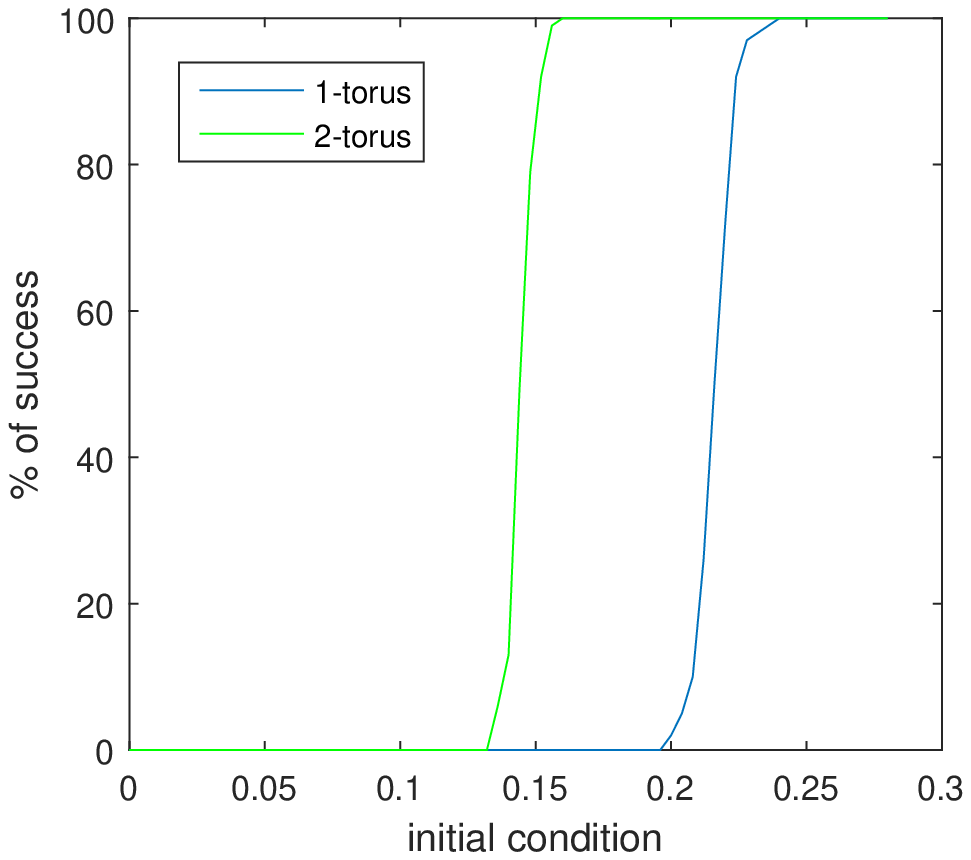}}
\caption{Simulations of the dynamics on $d$-tori with $N=1024$.}
\label{tori}
\end{figure}

Moreover it may be interesting to examine the simulations in Fig. \ref{tori}. From these simulations it seems that the connectivity of the graph might play an important role in the value of the various thresholds: a stronger persuasion strength (Fig. \ref{tori beta}) as well a higher initial condition (Fig. \ref{tori z0}) seem to be needed to insure the spreading of the asset in graphs with lower connectivity (in our case as $k$ decreases).

\section{Conclusions}
In this work, we deepened the analysis of the network dynamics for the diffusion of the adoption of a new technological item like a smartphone application or a PC program we proposed in \cite{mtns}. The main novelty of this dynamics is in the fact that the spread of such ``light choices'  is driven by a mechanism whose strength depends on the global diffusion of the item in the community, coupled with a spontaneous regression drift. 

In \cite{mtns}, we also analyzed the behavior of the system in the very simple case of a complete graph. In particular, in Theorem \ref{absorbing}, we exposed the presence of a phase transition depending on the gossip strength parameter and we highlighted the existence of an intermediate regime where, the diffusion and the permanence of the asset in the community depends on the size of the initial fraction of population possessing that asset. The presence of this last regime was the main novelty of this model with respect to classical epidemic models and it is coherent both with other models with similar driving mechanism and with intuition. 

In this work we proved Theorem \ref{absorbing general}, extending this results to the case of a generic graph, relating the presence of the phase transition (and the intermediate regime) not only to the gossip strength parameter, but also to some features of the graphs such as the degrees, the spectral radius of the adjacency matrix and the bottleneck ratio. In particular, this extension provides a result similar to case of a complete graph for a large family of graphs, in which  ER graphs and random configuration models with fixed and bounded degree distribution are included. 

Finally, we presented some applications to given topologies, both analytical and via numerical simulations, showing that such phenomena seems to be present even in other network topologies, which are not covered by the analytical results of Theorem \ref{absorbing general} (e.g. scale-free graphs, regular grids).

\end{document}